\documentclass[english,prl,twocolumn]{revtex4}
\usepackage{mathptmx}
\usepackage[T1]{fontenc}
\usepackage[latin9]{inputenc}
\usepackage{textcomp}
\usepackage{amsthm}
\usepackage{amsmath}
\usepackage{amssymb}

\makeatletter
\@ifundefined{textcolor}{}
{%
 \definecolor{BLACK}{gray}{0}
 \definecolor{WHITE}{gray}{1}
 \definecolor{RED}{rgb}{1,0,0}
 \definecolor{GREEN}{rgb}{0,1,0}
 \definecolor{BLUE}{rgb}{0,0,1}
 \definecolor{CYAN}{cmyk}{1,0,0,0}
 \definecolor{MAGENTA}{cmyk}{0,1,0,0}
 \definecolor{YELLOW}{cmyk}{0,0,1,0}
 }
\theoremstyle{plain}

\makeatother

\newtheorem{thm}{Theorem}  
 \newtheorem{lem}{Lemma}
\newtheorem{prop}{Proposition} 

\makeatother

\usepackage{babel}

\begin{document}

\title{{\Large Operational measure of incompatibility of noncommuting observables}}

\author{Somshubhro Bandyopadhyay}

\affiliation{Department of Physics and Center for Astroparticle Physics and Space
Science, Bose Institute, Block EN, Sector V, Bidhan Nagar, Kolkata
700091, India }

\email{email: som@bosemain.boseinst.ac.in}

\author{Prabha Mandayam }

\affiliation{The Institute of Mathematical Sciences, C. I. T. Campus, Taramani,
Chennai 600113, India }

\email{email: prabhamd@imsc.res.in}
\begin{abstract}
Uncertainty relations are often considered to be a measure of incompatibility
of noncommuting observables. However, such a consideration is not
valid in general, motivating the need for an alternate measure that
applies to any set of noncommuting observables. We present an operational
approach to quantifying incompatibility without invoking uncertainty
relations. Our measure aims to capture the incompatibility of noncommuting
observables as manifest in the nonorthogonality of their eigenstates.
We prove that this measure has all the desired properties. It is zero
when the observables commute, strictly greater than zero when they
do not, and is maximum when they are mutually unbiased. We also obtain
tight upper bounds on this measure for any $N$ noncommuting observables
and compute it exactly when the observables are mutually unbiased. 
\end{abstract}
\maketitle
In quantum theory, any observable or a set of commuting observables
can in principle be measured with any desired precision. This is because
commuting observables have a complete set of simultaneous eigenkets,
and therefore, measurement of one does not disturb the measurement
result obtained for the other. This no longer holds when the observables
do not commute. Noncommuting observables do not have a complete set
of common eigenkets, and therefore it is impossible to specify definite
values simultaneously. This is the essence of the celebrated uncertainty
principle \cite{heisenberg,Robertson,note-robertson}. Uncertainty
relations \cite{Deutsch,heisenberg,Masen-Uffnik,Robertson,beckner,birula-1,birula-2,Hayden,Hirschmann,Ivanovic,Prabha,sanchez,sanchez-R-1,sanchez-R-2,Wehner-Winter,wehner-winter-II}
express the uncertainty principle in a quantitative way by providing
a lower bound on the {}``uncertainty'' in the result of a simultaneous
measurement of noncommuting observables. 

Observables are defined to be compatible when they commute, and incompatible
when they do not. The uncertainty principle, therefore, is a manifestation
of the incompatibility of noncommuting observables. Despite the conceptual
importance of incompatible observables and applications of such observables
in quantum state determination \cite{Ivanovic-MUB,MUB-Durt,Wootters-fields-MUB}
and quantum cryptography \cite{BB84,Bruss-crypto,Cerf,HBP +Peres,HBP+tittel},
there does not seem to be a good general measure of their incompatibility,
although entropic uncertainty relations have often been considered
for this purpose (see, for example, \cite{Maassen,Azarchs,ballester-Wehner,Wehner-Winter}). 

To see in what sense uncertainty relations quantify incompatibility
of noncommuting observables, consider, for example, the entropic uncertainty
relation due to Maassen and Uffink \cite{Masen-Uffnik}. For any quantum
state $\rho\in\mathcal{H}$ with $\dim\mathcal{H}=d$, and measurement
of any two observables $A$ and $B$ with eigenvectors $\left\{ |a_{i}\rangle\right\} $
and $\left\{ |b_{i}\rangle\right\} $, respectively, it was shown
that \cite{Masen-Uffnik} \begin{eqnarray}
\frac{1}{2}\left(H\left(A|\rho\right)+H\left(B|\rho\right)\right) & \geq & -\log c,\label{MU}\end{eqnarray}
where $c=\max\left|\left\langle a|b\right\rangle \right|$: $|a\rangle\in\left\{ |a_{i}\rangle\right\} ,|b\rangle\in\left\{ |b_{i}\rangle\right\} $,
and $H\left(X|\rho\right)=-\sum_{i=1}^{d}\left\langle x_{i}\left|\rho\right|x_{i}\right\rangle \log\left\langle x_{i}\left|\rho\right|x_{i}\right\rangle $
for $X\in\left\{ A,B\right\} $ is the Shannon entropy (all logarithms
are taken to base 2). Observe that the right-hand side of the above
inequality is independent of $\rho$. The incompatibility of the observables
$A$ and $B$ can be measured by either the sum of the entropies {[}left
hand side of (\ref{MU}){]} minimized over all $\rho$ (if it is not
known whether equality is achieved) or the lower bound when equality
is achieved for some state. We then say that a set of observables
is more incompatible than another if the sum (or the lower bound)
takes on a larger value. It is clear from the above inequality that
a pair of observables is most incompatible when the observables are
mutually unbiased. Incompatibility of more than two observables can
be similarly quantified via a generalized form of the inequality (\ref{MU})
\cite{Wehner-Winter} when such an inequality can be found.

However, it is easy to see why inequality (\ref{MU}) is not a satisfactory
measure of incompatibility for \emph{all} pairs of incompatible observables.
This is because both sides of the inequality could be zero even when
the observables do not commute. This happens, for example, when the
noncommuting observables $A$ and $B$ are such that they commute
on a subspace. Such observables have one or more common eigenvectors
but not all eigenvectors are common because the observables do not
commute. For such a pair of observables both sides of inequality (\ref{MU})
are identically zero even though the observables are known to be incompatible.
Thus, uncertainty relations can only quantify incompatibility when
the observables do not have any common eigenvector. This shows that
uncertainty relations cannot be considered as a valid measure of incompatibility
for \emph{all} sets of noncommuting observables, thus motivating the
present work. Furthermore, incompatibility of more than two observables
is much less understood because uncertainty relations (in cases where
they are indeed a good measure) are known only for some special classes
of observables \cite{Hayden,sanchez,sanchez-R-1,sanchez-R-2,wehner-winter-II}.
Even for these cases maximally tight uncertainty relations are not
always known to exist \cite{Wehner-Winter,Prabha}. 

In this work, we present an operational approach to quantifying the
incompatibility of any set of $N$ noncommuting observables. We first
observe that, by definition, noncommuting observables do not have
a complete set of common eigenkets. Therefore, some of the eigenstates,
if not all, corresponding to different noncommuting observables must
be nonorthogonal. We therefore suggest a measure that quantifies incompatibility
of the observables as manifest in the nonorthogonality of their eigenstates.
We show that our measure applies to any set of noncommuting observables
(even if the observables commute on a subspace) and has the following
desirable properties. It is zero when the observables commute, strictly
greater than zero when they do not (note that the approach based on
an uncertainty relation fails in this regard), and maximum when they
are mutually unbiased. We also obtain nontrivial upper bounds for
any $N$ noncommuting observables, and show that they are tight when
$N\leq d+1$. We prove the latter by computing the measure exactly
for any $N$ mutually unbiased observables. 

In order to define our measure of incompatibility, we adopt the following
operational approach, best understood in the setting of quantum cryptography.
We imagine a quantum key distribution (QKD) protocol between two observers,
say, Alice and Bob, in presence of an eavesdropper employing an intercept-resend
attack. Alice transmits quantum states drawn randomly from an ensemble
(signal ensemble) $S$ of equiprobable pure states, where the pure
states are taken to be the eigenstates of the noncommuting observables
whose incompatibility we wish to quantify. The eavesdropper performs
a fixed measurement on every intercepted state (we assume that all
transmitted states are intercepted), replaces the original state with
some other state based on the measurement outcome, and sends it on
to Bob. Our measure is defined as the complement of the accessible
fidelity \cite{Fuchs=Sasaki,Fuchs} (the best possible average fidelity
an eavesdropper can obtain) of the set $S$. Intuitively, this measure
corresponds to the {}``amount of information'' that is \emph{inaccessible}
to an eavesdropper. 

For any given set $\Pi=\left\{ \Pi^{1},\Pi^{2},...,\Pi^{N}\right\} $
of $N$ noncommuting observables acting on a Hilbert space $\mathcal{H}_{d}$
of dimension $d$, the signal ensemble is defined as a set of pure
states $S\left(\Pi\right)=\left\{ \Pi_{j}^{i}=|\psi_{j}^{i}\rangle\langle\psi_{j}^{i}|\right\} $,
with $i=1,...,N$ and $j=1,...,d$, where $|\psi_{j}^{i}\rangle$
is the $j^{th}$ eigenvector of the observable $\Pi^{i}$. As explained
before, Alice transmits pure states $\Pi_{j}^{i}$ drawn randomly
from the set $S\left(\Pi\right)$ (probability of each state being
equal to $1/Nd$) in the presence of an eavesdropper employing an
intercept-resend strategy comprising of some measurement (POVM) $\mathbf{M}$
and a state reconstruction map $\mathbf{A}$. First we define the
notions of average and accessible fidelity (see Refs. \cite{Fuchs=Sasaki,Fuchs}
for definitions in a more general scenario). For a measurement $\mathbf{M}=\left\{ M_{a}\right\} $,
and a state reconstruction procedure $\mathbf{A}:a\rightarrow\sigma_{a}$
such that when the measurement outcome is $a$, the eavesdropper substitutes
the intercepted state with the state $\sigma_{a}$ and sends this
state to Bob, the \emph{average fidelity} of $S\left(\Pi\right)$
is given by:\begin{eqnarray}
F_{S\left(\Pi\right)}\left(\mathbf{M},\mathbf{A}\right) & = & \frac{1}{Nd}\sum_{ija}\mbox{Tr}\left(\Pi_{j}^{i}M_{a}\right)\mbox{Tr}\left(\Pi_{j}^{i}\sigma_{a}\right),\label{F-avg-N}\end{eqnarray}
where $\frac{1}{Nd}\mbox{Tr}\left(\Pi_{j}^{i}M_{a}\right)$ is the
joint probability for the state $\Pi_{j}^{i}$ and outcome $a$ of
the measurement, and $\mbox{Tr}\left(\Pi_{j}^{i}\sigma_{a}\right)$
is the fidelity achieved in this case. The \emph{optimal fidelity}
is obtained by maximizing the average fidelity over all measurements
and state reconstruction procedures: \begin{eqnarray}
F_{S\left(\Pi\right)} & = & \sup_{\mathbf{M}}\sup_{\mathbf{A}}\frac{1}{Nd}\sum_{ija}\mbox{Tr}\left(\Pi_{j}^{i}M_{a}\right)\mbox{Tr}\left(\Pi_{j}^{i}\sigma_{a}\right).\label{F-optimal}\end{eqnarray}
The optimal fidelity represents the best possible average fidelity
an eavesdropper can obtain. The measure of incompatibility of the
noncommuting observables in the set $\Pi$ is now defined as\emph{
\begin{eqnarray}
Q\left(\Pi\right) & = & 1-F_{S\left(\Pi\right)}.\label{Q-def}\end{eqnarray}
}It is clear from the definition that the measure is applicable even
when the noncommuting observables $\left\{ \Pi^{i}\right\} $ have
one or more common eigenvectors. We will say that a set of observables
$\Pi_{1}$ is more incompatible than another, say, $\Pi_{2}$, if
the former takes on a larger $Q$ value. It is interesting to note
that the comparison holds regardless of the number of observables
in each set and the dimension of the Hilbert space. 

For any set $\Pi$ of $N$ noncommuting observables, $Q\left(\Pi\right)$
can in principle be computed but requires optimization which may be
difficult to perform in general. Nevertheless, we give a simplified
expression of a closely related quantity which might be useful in
computing the measure for special classes of observables. We further
note that our formalism is completely general in the sense that it
can be applied to observables not all of which are commuting. Suppose
we have a set $\mathfrak{S}$ of $\mathfrak{N}$ observables, in which
some observables do not commute. From such a set one can always construct
a minimal subset $S$ of $N\leq\mathfrak{N}$ noncommuting observables
with the property that any observable that is not in $S$ must commute
with at least one observable in $S$. For example, if $N=1$, then
it means that all observables in $\mathfrak{S}$ commute with each
other, whereas $N=\mathfrak{N}$ implies that all observables in $\mathfrak{S}$
are noncommuting. Incompatibility of any set of observables is then
defined as the incompatibility of the minimal noncommuting set obtained
in this fashion. 

The remainder of the paper is arranged as follows. We begin by proving
two basic properties of $Q$ (Proposition 1) and obtain the upper
bounds (Theorem 1). We will then derive a simplified expression of
a quantity closely related to optimal fidelity (Theorem 2) and use
it to compute $Q\left(\Pi\right)$ exactly for any $N$ mutually unbiased
observables (Theorem 3). The result in Theorem 3 implies that the
upper bounds in Theorem 1 are tight. Finally we conclude with implications
of these results in quantum cryptography and suggest future directions
of research. 

\begin{prop} $Q=0$ for commuting observables and $Q>0$ when the
observables do not commute. \end{prop}
\begin{proof}
If the observables commute, then they have a complete set of common
eigenkets which form an orthonormal basis. Thus, the minimal noncommuting
set $\Pi$ has only one element (any member of the commuting set),
i.e., $N=1$ and the set $S\left(\Pi\right)$ consists only of the
common eigenkets which are mutually orthogonal. This implies that
the optimal fidelity as defined by Eq. (\ref{F-optimal}) is $1$,
and therefore $Q=0$.

When the observables do not commute, then the minimal noncommuting
set $\Pi$ has at least two noncommuting observables. Then some of
the eigenstates in $S\left(\Pi\right)$, if not all, belonging to
different noncommuting observables must be nonorthogonal. Because
non-orthogonal states cannot be distinguished perfectly, we have $F_{S\left(\Pi\right)}<1$,
and therefore, $Q>0$. This completes the proof.
\end{proof}
We now obtain upper bounds on $Q\left(\Pi\right)$. The bounds are
tight for mutually unbiased observables as shown in Theorem 3. 

\begin{thm} The following bounds hold for a set $\Pi$ of $N$ noncommuting
observables acting on $\mathcal{H}_{d}$ with $\dim\mathcal{H}_{d}=d$:
\begin{eqnarray}
Q\left(\Pi\right) & \leq & \left(1-\frac{1}{N}\right)\left(1-\frac{1}{d}\right),\quad N\leq d+1\label{thm-1-eq-1}\\
Q\left(\Pi\right) & \leq & \frac{d-1}{d+1},\quad N\geq d+1\label{thm-1-eq-2}\end{eqnarray}
 \end{thm} 

Before we get to the proof, we would like to point out that both bounds
hold for any $N$. However, they are competing in the sense that one
is better than the other depending on whether $N<d+1$ or $N>d+1$,
and are equal when $N=d+1$. 
\begin{proof}
We first prove inequality (\ref{thm-1-eq-1}). We pick a measurement
$\mathbf{M}$ and a state reconstruction strategy $\mathbf{A}$ to
obtain a lower bound on the average fidelity {[}Eq. (\ref{F-avg-N}){]};
the result then follows from the definitions of optimal fidelity and
$Q\left(\Pi\right)$. The measurement $\mathbf{M}$ that we choose
is the standard von Neumann measurement in the eigenbasis of some
observable $\Pi^{k}\in\Pi$.  Thus, the measurement $\mathbf{M}=\left\{ \Pi_{l}^{k}=|\psi_{l}^{k}\rangle\langle\psi_{l}^{k}|\right\} _{l=1}^{d}$
consists of rank-1 orthogonal projection operators satisfying $\mbox{Tr}\left(\Pi_{j}^{k}\Pi_{l}^{k}\right)=\delta_{jl}$
and $\sum_{l}\Pi_{l}^{k}=\mathbb{I}.$ The state reconstruction map
$\mathbf{A}$ reproduces the state $\Pi_{l}^{k}$ if the outcome is
$l$. With this, one can show that (details in the Appendix) \begin{eqnarray}
F_{S\left(\Pi\right)}\left(\mathbf{M},\mathbf{A}\right) & \geq & \frac{N+d-1}{Nd}.\label{F-avg-proof-thm-1}\end{eqnarray}
Noting that, by definition, $F_{S\left(\Pi\right)}\geq F_{S\left(\Pi\right)}\left(\mathbf{M},\mathbf{A}\right)$,
we get \begin{eqnarray}
F_{S(\left(\Pi\right)} & \geq & \frac{N+d-1}{Nd}.\label{f-opt-lower-bound}\end{eqnarray}
Inequality (\ref{thm-1-eq-1}) now follows from the definition of
$Q\left(\Pi\right)$. To prove the upper bound in (\ref{thm-1-eq-2})
we simply use a lower bound on the best possible average fidelity
(accessible fidelity in the terminology of \cite{Fuchs=Sasaki,Fuchs})
obtained for any pure state ensemble $\mathcal{E}=\left\{ |\psi_{i}\rangle,p_{i}\right\} $
\cite{Fuchs}, \begin{eqnarray*}
F_{\mathcal{E}} & \geq & \frac{2}{d+1},\end{eqnarray*}
from which the result follows by definition of $Q\left(\Pi\right)$. 
\end{proof}
Ideally, we would like to compute $Q\left(\Pi\right)$ for any set
$\Pi$. Unfortunately, there is no straightforward way to do the optimization
in Eq. (\ref{F-optimal}). Nevertheless, we hope to get some insight
to the problem by obtaining a simplified form of the so-called\emph{
}achievable fidelity \cite{Fuchs,Fuchs=Sasaki} obtained by maximizing
the average fidelity over all state reconstruction strategies: \begin{eqnarray}
F_{S\left(\Pi\right)}\left(\mathbf{M}\right) & = & \sup_{\mathbf{A}}\frac{1}{Nd}\sum_{ija}\mbox{Tr}\left(\Pi_{j}^{i}M_{a}\right)\mbox{Tr}\left(\Pi_{j}^{i}\sigma_{a}\right).\label{F-achievable}\end{eqnarray}
As one can easily see, the optimal fidelity {[}Eq. (\ref{F-optimal}){]}
can now be expressed as \begin{eqnarray}
F_{S\left(\Pi\right)} & = & \sup_{\mathbf{M}}F_{S\left(\Pi\right)}\left(\mathbf{M}\right).\label{F-accessible}\end{eqnarray}
We will assume, without any loss of generality that the POVM $\mathbf{M}=\left\{ M_{a}\right\} $
consists only of rank one elements: $M_{a}=m_{a}\chi_{a},$ where
$\chi_{a}=|\chi_{a}\rangle\langle\chi_{a}|$ is the density matrix
corresponding to the normalized vector $|\chi_{a}\rangle$. For any
such measurement one can calculate the achievable fidelity explicitly
\cite{Fuchs=Sasaki}: \begin{eqnarray}
F_{S\left(\Pi\right)}\left(\mathbf{M}\right) & = & \sum_{a}m_{a}\lambda\left(\Phi\left(\chi_{a}\right)\right)\label{f-ach-1}\end{eqnarray}
where $\Phi$ is a trace non-increasing completely positive linear
map whose action on any density matrix $\varrho$ is given by \begin{eqnarray}
\Phi\left(\varrho\right) & = & \frac{1}{Nd}\sum_{i}\Pi_{j}^{i}\varrho\Pi_{j}^{i}\label{PHI}\end{eqnarray}
and $\lambda\left(\Phi\left(\varrho\right)\right)$ is the largest
eigenvalue of the Hermitian operator $\Phi\left(\varrho\right)$. 

\begin{thm} For any $S\left(\Pi\right)$, and a measurement $\mathbf{M}=\left\{ M_{a}=m_{a}\chi_{a}\right\} $
the achievable fidelity is given by \begin{eqnarray}
F_{S\left(\Pi\right)}\left(\mathbf{M}\right) & = & \frac{1}{Nd}\sum_{ij}p\left(a\right)_{j}^{i}q\left(a\right)_{j}^{i}\label{f-ach-3}\end{eqnarray}
where $p\left(a\right)_{j}^{i}=\mbox{Tr}\left(\Pi_{j}^{i}\chi_{a}\right)$,
$q\left(a\right)_{j}^{i}=\langle\eta_{a}|\Pi_{j}^{i}|\eta_{a}\rangle$,
and $|\eta_{a}\rangle$ is the eigenvector corresponding to the largest
eigenvalue of $d\Phi\left(\chi_{a}\right)$. \end{thm} 
\begin{proof}
Using Eq. (\ref{PHI}) we can write $\Phi\left(\chi_{a}\right)$ as
\begin{eqnarray}
\Phi\left(\chi_{a}\right) & = & \frac{1}{Nd}\sum_{ij}\Pi_{j}^{i}\chi_{a}\Pi_{j}^{i}\nonumber \\
 & = & \frac{1}{Nd}\sum_{ij}\mbox{Tr}\left(\Pi_{j}^{i}\chi_{a}\right)\Pi_{j}^{i}\label{phi-chi}\end{eqnarray}
Observe that $d\Phi\left(\chi_{a}\right)$ is a density matrix. Let
us call it $\rho\left(\Phi,\chi_{a}\right)$, and the probabilities
$\mbox{Tr}\left(\Pi_{j}^{i}\chi_{a}\right)=p\left(a\right)_{j}^{i}$.
Thus, \begin{eqnarray}
\rho\left(\Phi,\chi_{a}\right) & = & \frac{1}{N}\left(\rho_{1}+\rho_{2}+\cdots+\rho_{N}\right),\label{rho-chi}\end{eqnarray}
where $\rho_{i}=\sum_{j=1}^{d}p\left(a\right)_{j}^{i}\Pi_{j}^{i}$.
Now suppose that $|\eta_{a}\rangle$ is the eigenvector of $\rho\left(\Phi,\chi_{a}\right)$
corresponding to the largest eigenvalue $\mu_{a}$. Then,\begin{eqnarray*}
\mu_{a} & = & \langle\eta_{a}|\rho\left(\Phi,\chi_{a}\right)|\eta_{a}\rangle\\
 & = & \frac{1}{N}\sum_{ij}p\left(a\right)_{j}^{i}q\left(a\right)_{j}^{i}\end{eqnarray*}
where $q\left(a\right)_{j}^{i}=\langle\eta_{a}|\Pi_{j}^{i}|\eta_{a}\rangle$.
Noting that $d\Phi\left(\chi_{a}\right)=\rho\left(\Phi,\chi_{a}\right)$,
the result follows from (\ref{f-ach-1}). 
\end{proof}
We now show that the upper bounds in Theorem 1 are tight for mutually
unbiased observables. Mutually unbiased observables are those observables
whose eigenvectors form mutually unbiased bases \cite{Bandyopadhyay-MUB,Ivanovic-MUB,MUB-Durt}.
For $N$ mutually unbiased observables, $\Pi^{1},\Pi^{2},...,\Pi^{N}$,
their eigenvectors satisfy:\begin{eqnarray}
\mbox{Tr}\left(\Pi_{j}^{i}\Pi_{k}^{i}\right) & = & \delta_{jk}\label{MUB-1}\\
\mbox{Tr}\left(\Pi_{j}^{i}\Pi_{l}^{k}\right) & = & \frac{1}{d}\;\;\mbox{when}\;\; i\neq k.\label{MUB-2}\end{eqnarray}
It is known that a complete set of $d+1$ mutually unbiased bases
exist in prime and prime powered dimensions \cite{Ivanovic-MUB,Wootters-fields-MUB,Bandyopadhyay-MUB}
. For other dimensions, however, the problem remains open. 

\begin{thm} Let $\Pi=\left\{ \Pi^{1},\Pi^{2},...,\Pi^{N}\right\} $
be a set of $N\leq d+1$ mutually unbiased observables acting on $\mathcal{H}_{d}$
with $\dim\mathcal{H}_{d}=d$. Then, \begin{eqnarray}
Q\left(\Pi\right) & = & \left(1-\frac{1}{N}\right)\left(1-\frac{1}{d}\right)\label{thm-2}\end{eqnarray}
\end{thm}
\begin{proof}
In this case $S\left(\Pi\right)=\left\{ \Pi_{j}^{i}\right\} $, with
$i=1,...,N$ and $j=1,...,d$, and the states $\Pi_{j}^{i}$ satisfy
Eqs. (\ref{MUB-1}) and (\ref{MUB-2}). Now Theorem 2 gives us an
exact expression for the achievable fidelity for any set $S\left(\Pi\right)$
and a measurement $\mathbf{M}=m_{a}\chi_{a}$. By applying the Schwartz
inequality to Eq. (\ref{f-ach-3}) one immediately obtains the following
bound on the achievable fidelity: \begin{eqnarray}
F_{S\left(\Pi\right)}\left(\mathbf{M}\right) & \leq & \frac{1}{Nd}\sum_{a}m_{a}\sqrt{\sum_{ij}\left(p\left(a\right)_{j}^{i}\right)^{2}}\sqrt{\sum_{ij}\left(q\left(a\right)_{j}^{i}\right)^{2}}\label{f-ach-3-2}\end{eqnarray}

We now use the following lemma, proof of which is given in the Appendix.

\begin{lem} Let $|\phi\rangle\in\mathcal{H}_{d}$. Let $\left\{ \Pi^{i},i=1,...,N\right\} $,
where $N\leq d+1$ be a set of mutually unbiased bases in $\mathcal{H}_{d}$.
Let $t_{j}^{i}=\langle\phi|\Pi_{j}^{i}|\phi\rangle$, where $\Pi_{j}^{i}$
is the $j^{th}$ vector of the $i^{th}$ basis. Then, \begin{eqnarray}
\sum_{i=1}^{N}\sum_{j=1}^{d}\left(t_{j}^{i}\right)^{2} & \leq & \frac{N+d-1}{d}\label{lem}\end{eqnarray}
\end{lem}

By application of Lemma 1 in inequality (\ref{f-ach-3-2}) we get,\begin{eqnarray}
F_{S\left(\Pi\right)}\left(\mathbf{M}\right) & \leq & \frac{N+d-1}{Nd^{2}}\sum_{a}m_{a}\nonumber \\
 & = & \frac{N+d-1}{Nd}\label{proof-thm-2-eq-3}\end{eqnarray}
where we have used $\sum_{a}m_{a}=d$ which follows from the fact
that the elements $\left\{ M_{a}\right\} $ of the POVM satisfy $\sum_{a}M_{a}=\mathbb{I}$.
Noting that the upper bound (\ref{proof-thm-2-eq-3}) holds irrespective
of the measurement $\mathbf{M}$, the optimal fidelity is therefore
bounded by \begin{eqnarray}
F_{S\left(\Pi\right)} & \leq & \frac{N+d-1}{Nd}.\label{proof-thm-2-eq-5}\end{eqnarray}
From the above inequality and the general lower bound on $F_{S\left(\Pi\right)}$
{[}inequality (\ref{f-opt-lower-bound}){]}, we therefore obtain ,
\begin{eqnarray}
F_{S\left(\Pi\right)} & = & \frac{N+d-1}{Nd}.\label{proof-2-eq-6}\end{eqnarray}
Eq. (\ref{thm-2}) now follows from the definition of $Q\left(\Pi\right)$. 
\end{proof}
In summary, we have pointed out that uncertainty relations cannot,
in general, be considered as a measure of incompatibility of noncommuting
observables. This observation led us to propose a measure of incompatibility
that applies to any set of noncommuting observables. The measure relies
on two simple facts: When observables do not commute, at least some
of their eigenstates must be nonorthogonal, and nonorthogonal quantum
states cannot be perfectly distinguished. The measure is shown to
satisfy the desired properties, namely, it is zero when the observables
commute and strictly greater than zero when they do not. We have also
obtained tight upper bounds for any $N$ noncommuting observables
and evaluated the measure exactly for mutually unbiased observables. 

We note that the underlying physical principle defining our measure
and the security of QKD protocols such as BB84 \cite{BB84} and its
generalizations \cite{Bruss-crypto,Cerf,HBP +Peres,HBP+Gisin,HBP+tittel}
is the same. Thus, the exact expression of incompatibility of any
$N$ mutually unbiased observables obtained here is expected to help
analyze the security of such protocols. We further note that, in recent
years entropic uncertainty relations have found applications in quantum
cryptography \cite{Damgaard,Koashi}, information locking \cite{Divincenzo},
and the separability problem \cite{Guhne}. We suspect that the results
presented here will also be useful in the aforementioned contexts.
Apart from these a recent result \cite{Spengler} shows that entanglement
can be detected by local mutually unbiased measurements. It is possible
that the results presented here could also help to obtain separability
bounds from incompatible measurements, other than mutually unbiased,
on the local subsystems. 

As a final comment, we feel that alternate measures of incompatibility
of observables should be explored for reasons outlined in the Introduction.
While this paper suggests only one such measure, the idea behind it
is quite general and it is likely that similar quantities might serve
as an equally good measure. Of course, it is hard to see how the difficulty
of general optimization could be avoided. \\

\emph{Acknowledgments:} SB thanks The Institute of Mathematical Sciences,
Chennai, for supporting his visit in June 2012, when part of this
work was completed. The authors are grateful to Bill Wootters for
his comments on an earlier version of this work.

\pagebreak

\section*{Appendix}

\subsection{Proof of Theorem 1.}

Here we will derive inequality (\ref{F-avg-proof-thm-1}). For the
choice of measurement $\mathbf{M}$ and the state reconstruction strategy
$\mathbf{A}$ discussed in the text, the average fidelity {[}Eq. (\ref{F-avg-N}){]}
is given by,\begin{eqnarray}
F_{S\left(\Pi\right)}\left(\mathbf{M},\mathbf{A}\right) & = & \frac{1}{Nd}\sum_{i,j,l}\left[\mbox{Tr}\left(\Pi_{j}^{i}\Pi_{l}^{k}\right)\right]^{2}\nonumber \\
 & = & \frac{1}{N}+\frac{1}{Nd}\sum_{i\neq k,j,l}\left|\langle\psi_{j}^{i}|\psi_{l}^{k}\rangle\right|^{4}.\label{eq:proof-thm-1-eq-2}\end{eqnarray}
Now the orthogonal states $\left\{ |\psi_{l}^{k}\rangle\right\} _{l=1}^{d}$
form a basis of the Hilbert space $\mathcal{H}_{d}$. Therefore, for
any state $|\psi_{j}^{i}\rangle$, \begin{eqnarray}
\sum_{l=1}^{d}\left|\langle\psi_{j}^{i}|\psi_{l}^{k}\rangle\right|^{2} & = & 1.\label{proof-thm-1-eq-4}\end{eqnarray}
Furthermore, for $k\neq i$, we can write \begin{eqnarray}
\left|\langle\psi_{j}^{i}|\psi_{l}^{k}\rangle\right|^{2} & = & \frac{1}{d}+\Delta_{jl}^{ik}\;:\; k\neq i,\label{proof-thm-1-eq-5}\end{eqnarray}
where $\Delta_{jl}^{ik}$ satisfies $-\frac{1}{d}\leq\Delta_{jl}^{ik}\leq1-\frac{1}{d}$.
From Eqs. (\ref{proof-thm-1-eq-5}) and (\ref{proof-thm-1-eq-4})
it follows that \begin{eqnarray}
\sum_{l=1}^{d}\Delta_{jl}^{ik} & = & 0\;:k\neq i.\label{proof-thm-1eq-8}\end{eqnarray}
Substituting (\ref{proof-thm-1-eq-5}) in Eq. (\ref{eq:proof-thm-1-eq-2})
and using Eq. (\ref{proof-thm-1eq-8}), one obtains\begin{eqnarray*}
F_{S\left(\Pi\right)}\left(\mathbf{M},\mathbf{A}\right) & = & \frac{1}{N}+\frac{1}{Nd}\sum_{i\neq k,j,l}\left(\frac{1}{d}+\Delta_{jl}^{ik}\right)^{2}\\
 & = & \frac{1}{N}+\frac{N-1}{Nd}+\frac{1}{Nd}\sum_{i\neq k,j,l}\left(\Delta_{jl}^{ik}\right)^{2}\\
 & \geq & \frac{1}{N}+\frac{N-1}{Nd}=\frac{\left(N+d-1\right)}{Nd}.\end{eqnarray*}

\subsection{Proof of Lemma 1}

The following lemma will help us to prove Lemma 1. 

\begin{lem} Let $V$ be a real vector space of dimension $\left(d-1\right)$
equipped with a inner product $\langle\psi|\phi\rangle=\sum_{k=1}^{d-1}a_{k}b_{k}$,
where $\left\{ a_{k}\right\} $ and $\left\{ b_{k}\right\} $ are
the components of the vectors $|\psi\rangle\in V$ and $|\phi\rangle\in V$
in some orthogonal basis. Let $\left\{ |\psi_{i}\rangle\right\} _{i=1}^{d}$
be a set of linearly dependent vectors spanning $V$, with the property
that $\langle\psi_{i}|\psi_{i}\rangle=\left(1-\frac{1}{d}\right)$
and $\langle\psi_{i}|\psi_{j}\rangle=-\frac{1}{d}$. Let $|\phi\rangle\in V$
such that $\sum_{i=1}^{d}\langle\psi_{i}|\phi\rangle=0$. Then $|\phi\rangle$
can be expressed as \begin{eqnarray*}
|\phi\rangle & = & \sum_{i=1}^{d}\lambda_{i}|\psi_{i}\rangle,\end{eqnarray*}
where $\lambda_{i}=\langle\phi|\psi_{i}\rangle$ are such that\begin{eqnarray*}
\langle\phi|\phi\rangle & = & \sum_{i=1}^{d}\lambda_{i}^{2}.\end{eqnarray*}
 \end{lem}
\begin{proof}
As the set of vectors $\left\{ |\psi_{i}\rangle\right\} _{i=1}^{d}$
span $V$ we can write $|\phi\rangle$ as:\begin{eqnarray}
|\phi\rangle & = & \sum_{i=1}^{d}\lambda_{i}|\psi_{i}\rangle.\label{lem-2-eq-1}\end{eqnarray}
By explicitly computing the inner product of $|\phi\rangle$ with
$|\psi_{k}\rangle$ we get, \begin{eqnarray}
\langle\psi_{k}|\phi\rangle & = & \lambda_{k}+\sum_{l=1}^{d}\lambda_{l}.\label{lem-2-eq-2}\end{eqnarray}
The given condition $\sum_{i=1}^{d}\langle\psi_{i}|\phi\rangle=0$,
together with the above equation gives \begin{eqnarray}
\left(d+1\right)\sum_{k=1}^{d}\lambda_{k} & = & 0,\label{lem-2-eq-3}\end{eqnarray}
 which implies that \begin{eqnarray}
\sum_{l=1}^{d}\lambda_{l} & = & 0.\label{eq:lem-2-3.1}\end{eqnarray}
Thus {[}from Eq. (\ref{lem-2-eq-2}){]}, \begin{eqnarray}
\langle\psi_{k}|\phi\rangle & = & \lambda_{k}.\label{lem-2-eq-4}\end{eqnarray}
 We will now prove that $\langle\phi|\phi\rangle=\sum_{i=1}^{d}\lambda_{i}^{2}$
by explicitly computing the squared norm:\begin{eqnarray*}
\langle\phi|\phi\rangle & = & \sum_{kj}\lambda_{k}\lambda_{j}\langle\psi_{j}|\psi_{k}\rangle\\
 & = & \left(1-\frac{1}{d}\right)\sum\lambda_{k}^{2}-\frac{1}{d}\sum_{j\neq k}\lambda_{k}\lambda_{j}\\
 & = & \sum\lambda_{k}^{2}.\end{eqnarray*}
where to arrive at the last line we have used Eq. (\ref{eq:lem-2-3.1}).
\end{proof}
To prove Lemma 1 we first note that the MUBs lie in the set of density
matrices which itself is a convex subset of complex $d\times d$ Hermitian
matrices. The set of complex $d\times d$ Hermitian matrices forms
an $d^{2}$ - dimensional real vector space $\mathcal{V}$ equipped
with the inner product $\mbox{Tr}\left(AB\right)$ for any two vectors
$A,B\in\mathcal{V}$. The density matrices are however, of unit trace
and non-negative, and therefore lie in an $\left(d^{2}-1\right)$
- dimensional subspace of $\mathcal{V}$. This subspace is nothing
but the vector space of all Hermitian matrices of unit trace.

For our purpose we will deal with the vector space of traceless Hermitian
matrices $\mathcal{W}$ of dimension $(d^{2}-1)$ with the inner product
defined as $\mbox{Tr}\left(AB\right)$ for any two traceless Hermitian
matrices $A,B\in\mathcal{W}$. Thus a density matrix $\rho$ is represented
by $\tilde{\rho}=\rho-\frac{\mathbb{I}}{d}$. It is easy to check
that the vectors belonging to different mutually unbiased bases are
now orthogonal: That is, \begin{eqnarray*}
\mbox{Tr}\left(\tilde{\Pi}_{j}^{i}\tilde{\Pi}_{l}^{k}\right) & = & \mbox{Tr}\left(\Pi_{j}^{i}-\frac{\mathbb{I}}{d}\right)\left(\Pi_{l}^{k}-\frac{\mathbb{I}}{d}\right)\\
 & = & 0\;\;\;\; i\neq k\end{eqnarray*}
Moreover, the vectors $\left\{ \tilde{\Pi}_{j}^{i}=\Pi_{j}^{i}-\frac{\mathbb{I}}{d}\right\} :j=1,...,d$
for a given $i$, span a $(d-1)$-dimensional subspace, say, $\mathcal{W}_{i}$.
Thus when $(d+1)$ mutually unbiased bases exist, the vector space
$\mathcal{W}$ can be decomposed into $(d+1)$ orthogonal subspaces;
that is, $\mathcal{W}=\mathcal{W}_{1}\oplus\cdots\oplus\mathcal{W}_{d+1}$.
Therefore, $\tilde{\rho}$ can be expressed as \begin{eqnarray*}
\tilde{\rho} & = & \sum_{i}\tilde{\rho}_{i},\end{eqnarray*}
where $\tilde{\rho}_{i}\in\mathcal{W}_{i}$. 

Let $r_{j}^{i}=\mbox{Tr}\left(\tilde{\rho}\tilde{\Pi}_{j}^{i}\right)$,
where the $r_{j}^{i}$s' determine the projection of $\tilde{\rho}$
onto the subspace $\mathcal{W}_{i}$ and are related to the probabilities
$t_{j}^{i}=\mbox{Tr}\left(\Pi_{j}^{i}\rho\right),j=1,...,d$, when
$\rho$ is measured in the $\Pi^{i}$ basis via the following relation:
\begin{eqnarray*}
r_{j}^{i} & = & t_{j}^{i}-\frac{1}{d}.\end{eqnarray*}
Note that $r_{j}^{i}$s' satisfy \begin{eqnarray*}
\sum_{j=1}^{d}r_{j}^{i} & = & 0,\end{eqnarray*}
by virtue of the fact that $\sum_{j}t_{j}^{i}=1$. Now observe that
$r_{j}^{i}=\mbox{Tr}\left(\tilde{\rho}\tilde{\Pi}_{j}^{i}\right)=\mbox{Tr}\left(\tilde{\rho}_{i}\tilde{\Pi}_{j}^{i}\right)$.
This follows from the facts that $\tilde{\rho}=\sum_{i}\tilde{\rho}_{i}$,
and $\mbox{Tr}\left(\tilde{\rho}_{k}\tilde{\Pi}_{j}^{i}\right)=0:\; k\neq i$.
Therefore, by Lemma 2 we can write $\tilde{\rho}_{i}$ as, \begin{eqnarray*}
\tilde{\rho}_{i} & = & \sum_{j}r_{j}^{i}\tilde{\Pi}_{j}^{i}.\end{eqnarray*}
Once again by Lemma 2, we obtain \begin{eqnarray*}
\mbox{Tr}\left(\tilde{\rho}_{i}^{2}\right) & = & \sum_{j}\left(r_{j}^{i}\right)^{2}.\end{eqnarray*}
Since $\mbox{Tr}\left(\tilde{\rho}\right)^{2}$ is simply the sum
of the squares of the lengths of the components of $\tilde{\rho}$
in the orthogonal subspaces $\mathcal{W}_{i}$, we have, \begin{eqnarray*}
\mbox{Tr}\left(\tilde{\rho}^{2}\right) & = & \sum_{ij}\left(r_{j}^{i}\right)^{2}.\end{eqnarray*}
For any density matrix $\rho$, the right hand side can be readily
evaluated:\begin{eqnarray}
\sum_{ij}\left(r_{j}^{i}\right)^{2} & = & \sum_{ij}\left(t_{j}^{i}\right)^{2}-\sum_{i}\frac{1}{d}\nonumber \\
 & = & \sum_{i}\left(\sum_{j}\left(t_{j}^{i}\right)^{2}-\frac{1}{d}\right).\label{proof-lemma-1}\end{eqnarray}
Observe that the quantity $\left(\sum_{j}\left(t_{j}^{i}\right)^{2}-\frac{1}{d}\right)=\sum_{j}\left(r_{j}^{i}\right)^{2}\geq0$. 

When $\rho$ corresponds to a pure state, then \begin{eqnarray*}
\mbox{Tr}\left(\tilde{\rho}^{2}\right) & = & \mbox{Tr}\left(\rho-\frac{\mathbb{I}}{d}\right)^{2}\\
 & = & 1-\frac{1}{d}.\end{eqnarray*}
To arrive at our result we simply note that when $i=1,...,N$, where
$N\leq d+1$, \begin{eqnarray}
\sum_{i=1}^{N}\sum_{j=1}^{d}\left(r_{j}^{i}\right)^{2} & \leq & 1-\frac{1}{d}.\label{lemma-1-proof}\end{eqnarray}
Using Eqs. (\ref{proof-lemma-1}) and (\ref{lemma-1-proof}) we therefore
obtain,\begin{eqnarray*}
\sum_{i=1}^{N}\sum_{j=1}^{d}\left(t_{j}^{i}\right)^{2} & \leq & \frac{N+d-1}{d}.\end{eqnarray*}
 The equality is reached only when the pure state lies in the union
of the subspaces $\mathcal{W}_{i}$, $i=1,...,N$. Also note that
when $N=d+1$ we get back the known result \cite{Larsen,klappenecker}.
This completes the proof of Lemma 1. 

\begin{thebibliography}{39}
\bibitem{heisenberg} W. Heisenberg, Über den anschaulichen Inhalt
der quantentheoretischen Kinematik und Mechanik (The actual content
of quantum theoretical kinematics and mechanics), Zeitschrift für
Physik , 43, 172 (1927). 

\bibitem{Robertson} H. P. Robertson, The uncertainty principle, Physical
Review 34, 163 (1929). 

\bibitem{Deutsch} D. Deutsch, Uncertainty in quantum measurements,
Physical Review Letters 50, 631 (1983). 

\bibitem{Masen-Uffnik} H. Maassen, and J. Uffink, Generalized entropic
uncertainty relations, Physical Review Letters 60, 1103 (1988). 

\bibitem{note-robertson} Robertson \cite{Robertson} generalized
Heisenberg's uncertainty relation \cite{heisenberg} for any two observables
$A$ and $B$ (in the chosen unit $\hbar=1$):\begin{eqnarray}
\Delta A\Delta B & \geq & \frac{1}{2}\left|\left\langle \psi\left|\left[A,B\right]\right|\psi\right\rangle \right|,\label{uncertainty relation}\end{eqnarray}
where $\Delta X=\sqrt{\left\langle \psi\left|X^{2}\right|\psi\right\rangle -\left\langle \psi\left|X\right|\psi\right\rangle ^{2}}$,
$X\in\left\{ A,B\right\} $ is the standard deviation resulting from
measuring $X$ on the quantum state $|\psi\rangle$. Deutsch pointed
out that the above inequality is in general too weak except for canonically
conjugate observables \cite{Deutsch}. 

\bibitem{Hirschmann} I. I. Hirschmann, A note on etntropy, American
Journal of Mathematics 79, 152 (1957). 

\bibitem{beckner} W. Beckner, Inequalities in Fourier analysis, Annals
of Mathematics 102, 159 (1975). 

\bibitem{birula-1} I. Bialynicki-Birula, and J. Mycielski, Uncertainty
relations for information entropy in wave mechanics, Communications
in Mathematical Physics 44, 129 (1975).

\bibitem{birula-2} I. Bialynicki-Birula, Entropic uncertainty relations,
Physics Letters A 103, 253 (1984). 

\bibitem{Ivanovic} I. D. Ivanovic, An inequality for the sum of entropies
of unbiased quantum measurements, Journal of Physics A: Math. Gen.
25, 363 (1992). 

\bibitem{sanchez} J. Sanchez, Entropic uncertainty and certainty
relations for complementary observables, Physics Letters A 173, 233
(1993). 

\bibitem{sanchez-R-1}J. Sanchez-Ruiz, Improved bounds in the entropic
uncertainty and certainty relations for complementary observables,
Physics Letters A 201, 125 (1995). 

\bibitem{sanchez-R-2} J. Sanchez-Ruiz, Optimal entropic uncertainty
relation in two-dimensional Hilbert space, Physics Letters A 244,
189 (1998). 

\bibitem{Prabha} P. Mandayam, S. Wehner, and N. Balachandran, A transform
of complementary aspects with applications to entropic uncertainty
relations, Journal of Mathematical Physics 51, 082201 (2010).

\bibitem{Wehner-Winter} S. Wehner, and A. Winter, Entropic uncertainty
relations\textemdash{}a survey, New Journal of Physics 12, 025009
(2010). 

\bibitem{wehner-winter-II} S. Wehner, and A. Winter, Higher entropic
uncertainty relations for anti-commuting observables, Journal of Mathematical
Physics 49, 062105 (2008). 

\bibitem{Maassen} H Maassen, A discrete entropic uncertainty relation,
Quantum probability and applications V, Lecture Notes in Mathematics
Volume 1442, 1990, pp 263-266 1990 - Springer.

\bibitem{ballester-Wehner} M. A. Ballester, and S. Wehner, Entropic
uncertainty relations and locking: Tight bounds for mutually unbiased
bases, Physical Review A 75, 022319 (2007).

\bibitem{Azarchs} A. Azarchs, Entropic uncertainty relations for
incomplete sets of mutually unbiased observables, arXiv preprint quant-ph/0412083,
(2004).

\bibitem{Hayden} P. Hayden, D. Leung, P. Shor, and A. Winter, Randomizing
quantum states: Constructions and applications, Communications in
Mathematical Physics 250, 371 (2004), IEEE New York. 

\bibitem{Damgaard} I. Damgaard, S. fehr, L. Salvail, and C. Schaffner,
Cryptography in the bounded storage model, Proceedings of 46th IEEE
FOCS, pp 449-458 (2005), 

\bibitem{Koashi} M. Koashi, Simple security proof of quantum key
distribution based on complementarity, New Journal of Physics 11,
045018 (2009). 

\bibitem{Divincenzo} D. DiVincenzo, M. Horodecki, D. Leung, J. Smolin,
and B. Terhal, Locking classical correlations in quantum states, Physical
Review Letters 92, 067902 (2004). 

\bibitem{Guhne}O. Gühne, Characterizing entanglement via uncertainty
relations, Physical Review Letters 92, 117903 (2004). 

\bibitem{Spengler} C. Spengler, M. Huber, S. Brierley, T. Adaktylos,
B. C. Hiesmayr, Entanglement detection via mutually unbiased bases,
Physical Review A 86, 022311 (2012).

\bibitem{BB84} C. H. Bennett, and G. Brassard, Quantum cryptography:
Public key distribution and coin tossing, Proceedings of the IEEE
International Conference on Computers, Systems and Signal Processing,
Bangalore, India, pp. 175-179 (1984), IEEE New York. 

\bibitem{Bruss-crypto}D. Bruss, Optimal eavesdropping in quantum
cryptography with six states, Physical Review Letters 81, 3018 (1998). 

\bibitem{HBP+Gisin} H. Bechmann-Pasquinucci, and N. Gisin, Incoherent
and coherent eavesdropping in the six-state protocol of quantum cryptography,
Physical Review A 59, 4238 (1999). 

\bibitem{HBP+tittel}H. Bechmann-Pasquinucci, and W. Tittel, Quantum
cryptography using larger alphabets, Physical Review A 61, 062308
(2000). 

\bibitem{HBP +Peres}H. Bechmann-Pasquinucci, and A. Peres, Quantum
cryptography with 3-state systems, Physical Review Letters 85, 3313
(2000). 

\bibitem{Cerf}N. J. Cerf, M. Bourennane, A. Karlsson, and N. Gisin,
Security of quantum key distribution using d-level systems, Physical
Review Letters 88, 127902 (2002). 

\bibitem{Fuchs=Sasaki} C. Fuchs, and M. Sasaki, Squeezing quantum
information through a classical channel: measuring the \textquotedbl{}quantumness\textquotedbl{}
of a set of quantum states, Quantum Information \& Computation 3,
377 (2003). 

\bibitem{Fuchs} C. Fuchs, On the quantumness of a Hilbert space,
Quantum Information \& Computation 4, 467 (2004). 

\bibitem{MUB-Durt} T. Durt, B. G. Englert, I. Bengtsson, K. \.{Z}yczkowski,
On mutually unbiased bases, International Journal of Quantum Information,
8, 535 (2010)

\bibitem{Ivanovic-MUB}I.D. Ivanovi\textasciiacute{}c, Geometrical
description of quantal state determination. Journal of Physics A 14,
3241 (1981). 

\bibitem{Wootters-fields-MUB}W.K. Wootters and B.D. Fields. Optimal
state-determination by mutually unbiased measurements. Ann. Physics,
191:363\textendash{}381, 1989.

\bibitem{Bandyopadhyay-MUB}S. Bandyopadhyay, P.O. Boykin, V. Roychowdhury,
and F. Vatan. A new proof of the existence of mutually unbiased bases.
Algorithmica, 34:512\textendash{}528, 2002.

\bibitem{Larsen} U. Larsen, Superspace geometry: the exact uncertainty
relationship between complementary aspects, Journal of Physics A:
Math. Gen. 23, 1041 (1990). 

\bibitem{klappenecker} A. Klappenecker, and M. Rotteler, Mutually
unbiased bases are complex projective 2-designs, Proceedings of the
2005 IEEE International Symposium on Information Theory (ISIT 05),
1740 (2005), IEEE New York. 

\end{thebibliography}
\end{document}